\documentclass[a4paper,11pt]{article}

\usepackage{authblk}
\usepackage[utf8]{inputenc}
\usepackage[T1]{fontenc}
\usepackage{pslatex}
\usepackage{algorithm}
\usepackage{algorithmic}
\usepackage{multirow}
\usepackage{graphicx}
\usepackage{booktabs}
\usepackage{balance}
\usepackage{epsfig}
\usepackage{epstopdf}
\usepackage{amsmath}
\usepackage{amsthm}
\usepackage{amsfonts}
\usepackage{subcaption}
\usepackage{wasysym}
\usepackage{microtype}
\usepackage{rotating}
\usepackage{txfonts}
\usepackage{url}
\usepackage{mathptmx}
\DeclareMathAlphabet{\mathcal}{OMS}{cmsy}{m}{n}
\usepackage[colorlinks,allcolors=blue]{hyperref}
\usepackage{xspace}
\usepackage{microtype}
\usepackage{amssymb}
\usepackage{mathrsfs}
\usepackage{marvosym}
\usepackage{xifthen}
\usepackage{tikz}
\usetikzlibrary{patterns}
\usetikzlibrary{positioning}
\usetikzlibrary{decorations.pathreplacing}
\tikzset{
	mybrace/.style={decorate,decoration={brace,aspect=#1}}
}

\newcommand{\N}{\mathbb{N}}

\newcommand{\F}{\mathbb{F}}

\newtheorem{definition}{Definition}

\newtheorem{proposition}{Proposition}
\newtheorem{theorem}{Theorem}

\newcommand{\set}[1]{\ifthenelse{\isempty{#1}}{\varnothing}{\{#1\}}}
\newcommand{\familyF}{\mathcal{F}}

\providecommand{\keywords}[1]{\textbf{\textit{Keywords }} #1}

\begin{document}

\title{Search Space Reduction of Asynchrony Immune Cellular Automata by Center Permutivity}

\author[1]{Luca Mariot}
\author[2]{Luca Manzoni}
\author[1]{Alberto Dennunzio}

\affil[1]{{\normalsize Dipartimento di Informatica, Sistemistica e
    Comunicazione, Universit\`{a} degli Studi di Milano-Bicocca, Viale Sarca
    336, 20126 Milano, Italy} \\

  {\small \texttt{\{luca.mariot, alberto.dennunzio\}@unimib.it}}}

\affil[2]{{\normalsize Dipartimento di Matematica e Geoscienze, Universit\`a
    degli Studi di Trieste, Via Valerio 12/1, 34127 Trieste, Italy} \\

  {\small \texttt{lmanzoni@units.it}}}

\maketitle

\begin{abstract}
We continue the study of asynchrony immunity in cellular automata (CA), which can be considered as a generalization of correlation immunity in the case of vectorial Boolean functions. The property could have applications as a countermeasure for side-channel attacks in CA-based cryptographic primitives, such as S-boxes and pseudorandom number generators. We first give some theoretical results on the properties that a CA rule must satisfy in order to meet asynchrony immunity, like central permutivity. Next, we perform an exhaustive search of all asynchrony immune CA rules of neighborhood size up to $5$, leveraging on the discovered theoretical properties to greatly reduce the size of the search space.
\end{abstract}


\section{Introduction}
In the last years, research about cryptographic applications of cellular automata (CA) focused on the properties of the underlying local rules~\cite{martin2006a,leporati2014a,formenti2014a}. In fact, designing a CA\nobreakdash-\hspace{0pt}based cryptographic primitive using local rules that are not highly nonlinear and correlation immune could make certain attacks more efficient.

The aim of this paper is to investigate a new property related to \emph{asynchronous CA} called \emph{asynchrony immunity} (AI), which could be of interest in the context of side-channel attacks. This property can be described by a three-move game between a user and an adversary. Let $\ell, r, m \in \N$, $n=m+\ell+r$ and $t \le m$. The game works as follows:
\begin{enumerate}
  \item The user chooses a local rule $f: \F_2^{\ell+r+1} \rightarrow \F_2$ of memory $\ell$ and anticipation $r$.
  \item The adversary chooses $j \le t$ cells of the CA in the range $\{0,\cdots,m-1\}$.
  \item The user evaluates the output distribution $D$ of the CA $F:\F_2^{n}\rightarrow\F_2^{m}$ and the distribution $\tilde{D}$ of the asynchronous CA $\tilde{F}: \F_2^{n} \rightarrow \F_2^{m}$ where the $j$ cells selected by the adversary are not updated.
  \item \emph{Outcome}: if both $D$ and $\tilde{D}$ equals the uniform distribution, the user wins. Otherwise, the adversary wins.
\end{enumerate}
A cellular automaton rule $f:\F_2^{\ell+r+1} \rightarrow \F_2$ is called $(t,n)$--asynchrony immune if, for every subset $I$ of at most $t$ cells both the asynchronous CA $\tilde{F}: \F_2^{n} \rightarrow \F_2^{m}$ resulting from not updating on the subset $I$ of cells and the corresponding synchronous CA $F:\F_2^{n} \rightarrow \F_2^{m}$ are balanced, that is, the cardinality of the counterimage of each $m$-bit configuration equals $2^{\ell+r}$. Thus, asynchrony immune CA rules represent the winning strategies of the user in the game described above.

Notice the difference between the asynchrony immunity game and the \emph{$t$-resilient functions game}~\cite{chor1985a}: in the latter, generic vectorial Boolean functions $F: \F_2^n \rightarrow \F_2^m$ are considered instead of cellular automata, and the adversary selects both values and positions of the $t$ input variables.

The side-channel attack model motivating our work is the following. Suppose that a CA of length $n$ is used as an S-box in a block cipher, and that an attacker is able to inject \emph{clock faults} by making $t$ cells not updating. If the CA is not $(t,n)$-AI, then the attacker could gain some information on the internal state of the cipher by analyzing the differences of the output distributions in the original CA and the asynchronous CA. Similar fault attacks have already been investigated on stream ciphers based on clock-controlled \emph{Linear Feedback Shift Registers} (LFSR), such as LILI-128~\cite{lili128}. For further information on the topic, Hoch and Shamir~\cite{hoch04} provide more references on clock fault attacks on stream ciphers.

This paper is an extended version of~\cite{mariot2016}. In particular, the new contribution is twofold: from the theoretical side, we formally prove the necessity of central permutivity to have asynchrony immunity, which was conjectured in~\cite{mariot2016} according to the experimental results reported there. From the empirical point of view, we employ this new theoretical result to consistently extend the experimental search of asynchrony immune rules, by considering larger neighborhood sizes.

In the remainder of this paper, we recall in Section~\ref{sec:basic-notions} the necessary basic notions about Boolean functions and (asynchronous) CA, and we formally introduce the definition of asynchrony immunity in Section~\ref{sec:definition-of-asynchrony-immunity}, giving some theoretical results regarding this property. In particular, we show that AI is invariant under the operations of reflection and complement and that, for high enough values of $t$ (the maximum number of blocked cells), \emph{central permutivity} is a necessary condition for asynchrony immunity. We then perform in Section~\ref{sec:exp} an exhaustive search of asynchrony immune CA having $8$ output cells and neighborhood size up to $5$, computing also their nonlinearity and algebraic normal form. Finally, we provide some possible ways to generalize the notion of asynchrony immunity and how this property can be linked to existing CA models in Section~\ref{sec:open-problems}, as well as pointing out other avenues for future research on the subject.

\section{Basic Notions}
\label{sec:basic-notions}
In this section, we cover all necessary background definitions about one-dimensional CA, Boolean functions, and vectorial Boolean functions. In particular, we refer the reader to~\cite{carlet2010a,carlet2010b} for an in-depth discussion of (vectorial) Boolean functions.

Recall that a \emph{Boolean function} is a mapping $f:\F_2^n \to \F_2$, where $\F_2 = \set{0,1}$ denotes the finite field of two elements. Once an ordering of the $n$-bit input vectors has been fixed, each Boolean function $f$ can be uniquely represented by the output column of its \emph{truth table}, which is a vector $\Omega_f$ of $2^n$ binary elements. Therefore, the set of all possible Boolean functions of $n$ variables, denoted by $\mathcal{B}_n$, has cardinality $2^{2^n}$. The interpretation of the vector $\Omega_f$ as a decimal number is also called the \emph{Wolfram code} of the function $f$. Another common way of representing a Boolean function is through its \emph{Algebraic Normal Form} (ANF), that is, as a sum of products over its input variables. More formally, given $f: \F_2^n \to \F_2$ and $x \in \F_2^n$, the ANF will be of the form
\begin{equation}
  \label{eq:anf}
  P_f(x) = \bigoplus_{I \in 2^{[n]}} a_I \left ( \prod_{i \in I} x_i \right ) \enspace ,
\end{equation}
where $[n]$ is the initial segment of the natural numbers determined by $n \in \N$, i.e., $[n] = \set{0, \ldots, n-1}$, and the set $I = \set{i_1, \ldots, i_t} \subseteq [n]$ is a subset of $t$ indices and thus an element of $2^{[n]}$, the power set of $[n]$. For all $I \in 2^{[n]}$ the coefficient $a_I \in \F_2$ is determined through the \emph{M\"{o}bius transform}~\cite{carlet2010a}. A function $f$ is called \emph{affine} if the only non null coefficients $a_I$ are such that $|I| \le 1$. In other words, the ANF is composed only of monomials of degree at most $1$.

Boolean functions used in the design of symmetric ciphers must satisfy a certain number of properties in order to withstand particular cryptanalytic attacks. Two of the most important properties are \emph{balancedness} and \emph{nonlinearity}. A Boolean function $f: \F_2^n \to \F_2$ is balanced if its output vector $\Omega_f$ is composed of an equal number of zeros and ones. Unbalanced Boolean functions produce a statistical bias in the output of a symmetric cipher, which can be exploited by an attacker.

The nonlinearity of $f$, on the other hand, is the minimum Hamming distance of $\Omega_f$ from the set of all affine functions. The value of nonlinearity of $f$ can be computed as $Nl(f) = 2^{-1}(2^{n} - W_{max}(f))$, where $W_{max}(f)$ is the maximum absolute value of the \emph{Walsh transform} of $f$~\cite{carlet2010a}. The nonlinearity of a Boolean function used in a cipher should be as high as possible, in order to thwart \emph{linear cryptanalysis} attacks. Nonetheless, there exist upper bounds on the nonlinearity achievable by a Boolean function with respect to the number of its input variables. In particular, for $n$ even it holds that $Nl(f) \le 2^{n-1} - 2^{\frac{n}{2}-1}$. Functions satisfying this bound with equality are called \emph{bent}. On the other hand, for $n$ odd the upper bound when $n \le 7$ is $Nl(f) \le 2^{n-1} - 2^{\frac{n-1}{2}}$, which is achieved by \emph{quadratic functions}. For $n > 7$, the exact bound is still not known.

Let $n,m \in \N$. A \emph{vectorial Boolean function} of $n$ input variables and $m$ output variables (also called an $(n, m)$-function) is a mapping $F: \F_2^n \to \F_2^m$. In particular, a $(n,m)$-function is defined by $m$ Boolean functions of the form $f_i: \F_2^n \to \F_2$, called \emph{coordinate functions}. Each $0 \le i < m$, each $f_i$ specifies the $i$-th output bit of $F$. That is, for each $x \in \F_2^n$, we have $F(x)_i = f_i(x)$ for $0 \le i < m$.

A \emph{one-dimensional cellular automaton} (CA) can be seen as a particular case of vectorial Boolean function by limiting the way the coordinate functions can be defined. Let $\ell, m, r \in \N$ be non-negative integers and let $n = \ell + m + r$. Let $f: \F_2^{\ell + r + 1} \to \F_2$ be a Boolean function of $\ell + r + 1$ variables. A \emph{cellular automaton} of length $n$ with local rule $f$, \emph{memory} $\ell$ and \emph{anticipation} $r$ is the $(n,m)$-function $F: \F_2^n \to \F_2^m$ defined for all $i \in \set{0, \ldots, m-1}$ and for all $x = (x_{-\ell}, \ldots, x_{m+r}) \in \F_2^n$ as:
\begin{equation}
  \label{eq:glob-rule-ca}
  F(x_{-\ell}, \ldots, x_{m+r-1})_i = f(x_{i-\ell}, \ldots, c_{i+r}).
\end{equation}
Thus, a CA is the special case of a vector Boolean function where all coordinate functions are defined uniformly.

A \emph{$t$-asynchronous CA}, or $t$-ACA, induced by $I$ is denoted by
$\tilde{F}_I$ and it is defined by the following global function
$\tilde{F}_I: \F_2^n \to \F_2^m$:
\begin{equation*}
  \tilde{F}_I(x_{-\ell},\ldots,x_{m+r-1})_i =
  \begin{cases}
    f_i(x_{i-\ell}, \ldots, x_{i+r}) & \text{if $i \notin I$}\\
    x_i & \text{if $i \in I$.}
  \end{cases}
\end{equation*}

We also recall that a local rule $f: \F_2^{\ell + r + 1} \to F_2$ is said to be \emph{center permutive} when for each $u \in \F_2^\ell$, $v \in \F_2^r$, and $y \in \F_2$ there exists a \emph{unique} $x \in \F_2$ such that $f(uyv) = x$. In the field $\F_2$, center permutivity can also be expressed in another way. A local rule $f: \F_2^{\ell + r} \to \F_2$ is center permutive if there exists a function $g: \F_2^{\ell + r} \to \F_2$ such that for all $x = (x_0, \ldots, x_{\ell + r}) \in \F_2^{\ell + r + 1}$ we have that:
\[
  f(x_0, \ldots, x_{\ell + r}) = x_{\ell} \oplus g(x_0, \ldots, x_{\ell - 1}, x_{\ell + 1}, \ldots, x_{\ell + r}) \;.
\]

\section{Definition of Asynchrony Immunity}
\label{sec:definition-of-asynchrony-immunity}

Recall that a CA $F: \F_2^n \to \F_2^m$ with $n = \ell + r + m$ is said to be
\emph{balanced} if for each $y \in \F_2^m$, the preimages of $y$, i.e., all
$x \in \F_2^n$ such that $F(x) = y$, denoted by $F^{-1}(y)$ is such that
$|F^{-1}(y)| = 2^{\ell + r}$. Asynchrony immune CA can then be defined as
follows:
\begin{definition}
  \label{def:ai-ca}
  Let $n, m, r, \ell, t \in \N$ be non-negative integers, with
  $n = \ell + m + r$, and $F: \F_2^n \to \F_2^m$ a balanced CA having local rule
  $f: \F_2^{\ell + r + 1} \to \F_2$.

  The CA $F$ is said to be \emph{$(t, n)$-asynchrony immune} (for short,
  $(t, n)$-AI) if for all sets $I \subseteq [m]$ with $|I| \le t$ the resulting
  $|I|$-ACA $\tilde{F}_I$ is balanced.
\end{definition}
Among all possible $2^{2^{\ell + r +1}}$ rules of memory $\ell$ and anticipation $r$, we are interested in finding local rules that generates asynchrony immune CA satisfying additional useful cryptographic properties, such as high nonlinearity. As a consequence, proving necessary conditions for a rule to generate a $(t, n)$-AI is useful in reducing the size of the search space.

We start by proving that, for large enough CA and for high enough values of $t$, a necessary condition of $f$ is central permutivity.

\begin{theorem}
  \label{thm:central-permutivity}
  Let $F: \F_2^n \to \F_2^m$ be a $(t, n)$-AI CA with memory $\ell$ and anticipation $r$. If $t \ge \ell + r$ and $n \ge 2\ell + 2r + 1$ then the local rule $f: \F_2^{\ell + r + 1} \to \F_2$ is \emph{center permutive}.
\end{theorem}
\begin{proof}
  Suppose $F$ to be $(t, n)$-AI with $t$ and $n$ as in the hypothesis. Let $y = u_1 a u_2v \in \F_2^m$ be a configuration with $u_1 \in \F_2^\ell$, $a \in \F_2$, $u_2 \in \F_2^r$, and $v \in \F_2^{m-\ell-r-1}$. Let the set $I \supseteq \set{0,\ldots,\ell, \ell + 2, \ldots, r}$ be a set of indices to be blocked. It then follows that each preimage of $y$ can be expressed in the form $x = w_1 u_1 b u_2 w_2$ with $w_1 \in \F_2^\ell$, $b \in \F_2$, and $w_2 \in \F_2^{m + r -\ell - 1}$. Notice that both $u_1$ and $u_2$ remain unchanged when applying $\tilde{F}_I$ to $x$, since their indices are all contained in $I$. This situation is illustrated in Fig.~\ref{fig:center-permutive}.

  Since the value of the cells in $w_1$ cannot influence \emph{any} cell in $\tilde{F}_I(x)$ (since all cells that can be influenced are blocked), if $x = w_1 u_1 b u_2 w_2$ is a preimage of $y$, also $x' = w_1' u_1 b u_2 w_2$ for every $w_1' \in \F_2^\ell$ is a preimage of $y$. Hence, the first $\ell$ cells of the automaton contribute a multiplicative factor of $2^\ell$ for the number of preimages.

  We are now going to prove that the remaining factor of $2^r$ for the number of preimages is entirely due to the last $m + r - 1$ cells (i.e., the part denoted by $w_2$).

  For the sake of argument, suppose that the multiplicative factor contributed by the last $m + r -\ell - 1$ cells (i.e., the part denoted by $w_2$ in the preimages) is less than $2^r$, since only a single other cell in the preimage can change (the one denoted by $b$), it follows that, in that case the following two configurations are preimages of $y$ for some choice of $w_2$:
  \begin{align*}
    && x & = w_1 u_1 0 u_2 w_2 \\
    && x'& = w_1 u_1 1 u_2 w_2\;.
  \end{align*}
  Notice that the value of $a$ in $y$ is either $0$ or $1$ and it is influenced only by its own value and the value of $u_1$ and $u_2$. Without loss of generality, suppose that $a = 0$. Consider now the preimages of $y' = u_1 1 u_2 v$. To obtain $1$ in the unblocked position between $u_1$ and $u_2$ then, it must be $f(u_1 0 u_2) = 1$ or $f(u_1 1 u_2) = 1$, but by our previous assumption, both $f(u_1 0 u_2)$ and $f(u_1 1 u_2)$ are equal to $0$, and $y'$ has no preimages. Hence, our hypothesis that the part denote by $w_2$ in the preimages contributes less than a factor of $2^r$ in the number or preimages is inconsistent with the fact that $\tilde{F}_I$ must be balanced.

  Therefore, the parts $w_1$ and $w_2$ contribute, respectively, factors $2^\ell$ and $2^r$ in the number of preimages, for a total of $2^{\ell + r}$ preimages. It follows that, for each $u_1 \in \F_2^\ell$, $u_2 \in \F_2^r$, and $a \in \F_2$ there should be only \emph{one} value $b \in \F_2$ such that $f(u_1 b u_2) = a$. This means that $f$ is center permutive. 
\end{proof}

\begin{figure*}
  \begin{center}
    \begin{tikzpicture}
      \draw (0,0) rectangle (1.5, 0.5);
      \node[anchor=center] at (0.75,0.25) {$w_1$};
      \draw (1.5,0) rectangle (3, 0.5);
      \node[anchor=center] at (2.25,0.25) {$u_1$};
      \draw (3,0) rectangle (3.5, 0.5);
      \node[anchor=center] at (3.25,0.25) {$b$};
      \draw (3.5,0) rectangle (5, 0.5);
      \node[anchor=center] at (4.25,0.25) {$u_2$};
      \draw (5,0) rectangle (9.5,0.5);
      \node[anchor=center] at (7.25,0.25) {$w_2$};

      \draw (1.5,-2) rectangle (3,-1.5);
      \node[anchor=center] at (2.25,-1.75) {$u_1$};
      \draw (3,-2) rectangle (3.5,-1.5);
      \node[anchor=center] at (3.25,-1.75) {$a$};
      \draw (3.5,-2) rectangle (5,-1.5);
      \node[anchor=center] at (4.25,-1.75) {$u_2$};
      \draw (5,-2) rectangle (8,-1.5);
      \node[anchor=center] at (6.5,-1.75) {$v$};

      \draw[dashed] (0,0) -- (1.5,-1.5);
      \draw[dashed] (9.5,0) -- (8,-1.5);

      \draw[pattern=north west lines] (1.5,0) rectangle (3,-1.5);
      \draw[pattern=north west lines] (3.5,0) rectangle (5,-1.5);
    \end{tikzpicture}
  \end{center}
  \caption{\label{fig:center-permutive} The construction employed by the proof
    of Theorem~\ref{thm:central-permutivity}. The patterned background denotes
    the blocked cells. Here is it is possible to see that the part labeled with
    $w_1$ cannot influence any of the output cells. The cell labeled $b$ can
    influence only the cell labeled $a$ in the output, thus forcing the local
    rule to be center permutive.}
\end{figure*}
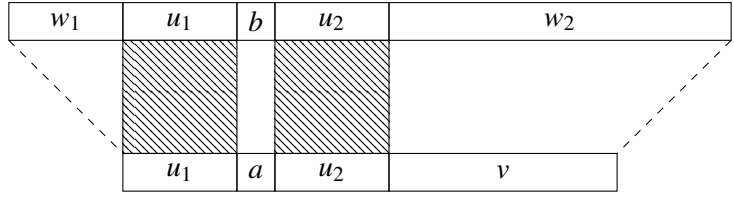

The previous theorem can be generalized as follows:

\begin{theorem}
  \label{thm:bijection}
  Let $F: \F_2^n \to \F_2^m$ be a $(t, n)$-AI CA with memory $\ell$ and
  anticipation $r$ and $k \in \N$ be a non-negative integer. Then, if
  $t \ge \ell + r$ and $n \ge 2\ell + 2r + k$, the function
  $F_{u,v}: \F_2^k \to \F_2^k$, which, for each $u \in \F_2^\ell$ and
  $v \in \F_2^r$, is defined as $F_{u,v}(x) = F'(uxv)$ where
  $F': \F_2^{k + \ell + r} \to \F_2^{k}$ is a CA with the same local rule as
  $F$, is a bijection.
\end{theorem}
\begin{proof}
  The proof of this theorem follows the same reasoning of the proof of
  Theorem~\ref{thm:central-permutivity}. Let $I$ be a set of indices to be
  blocked such that $I \supseteq \set{0,\ldots,\ell-1,\ell+k,\ell+k+r}$. Each
  element of $\F_2^m$ can then be rewritten in the form $y = u_1au_2v$ with
  $u_1 \in \F_2^\ell$, $u_2 \in \F_2^r$, $a \in \F_2^k$, and
  $v \in \F_2^{m -\ell -r -k}$. Similarly, a preimage of $y$ can be expressed in
  the form $x = w_1u_1au_2w_2$ with $w_1 \in \F_2^\ell$,
  $w_2 \in \F_2^{m+r-\ell-k}$, and $a \in \F_2^k$. Following the same reasoning
  of the proof of Theorem~\ref{thm:central-permutivity}, it can be shown that
  the $w_1$ part of the preimage contributes a factor $2^\ell$ in the number of
  preimages and that the $w_2$ part contributes a factor of $2^r$. Hence, the part
  denoted by $b$ in $y$ can have only one preimage. Therefore, when restricted
  to the $k$ cells ``surrounded'' by $u_1$ and $u_2$, the global function of the
  CA is a bijection, as desired.  
\end{proof}

\begin{figure*}
  \begin{center}
    \begin{tikzpicture}
      \draw (0,0) rectangle (1.5, 0.5);
      \node[anchor=center] at (0.75,0.25) {$w_1$};
      \draw (1.5,0) rectangle (3, 0.5);
      \node[anchor=center] at (2.25,0.25) {$u_1$};
      \draw (3,0) rectangle (4.5, 0.5);
      \node[anchor=center] at (3.75,0.25) {$b$};
      \draw (4.5,0) rectangle (6, 0.5);
      \node[anchor=center] at (5.25,0.25) {$u_2$};
      \draw (6,0) rectangle (10.5,0.5);
      \node[anchor=center] at (7.75,0.25) {$w_2$};

      \draw (1.5,-2) rectangle (3,-1.5);
      \node[anchor=center] at (2.25,-1.75) {$u_1$};
      \draw (3,-2) rectangle (4.5,-1.5);
      \node[anchor=center] at (3.75,-1.75) {$a$};
      \draw (4.5,-2) rectangle (6,-1.5);
      \node[anchor=center] at (5.25,-1.75) {$u_2$};
      \draw (6,-2) rectangle (9,-1.5);
      \node[anchor=center] at (7,-1.75) {$v$};

      \draw[dashed] (0,0) -- (1.5,-1.5);
      \draw[dashed] (10.5,0) -- (9,-1.5);

      \draw[pattern=north west lines] (1.5,0) rectangle (3,-1.5);
      \draw[pattern=north west lines] (4.5,0) rectangle (6,-1.5);

      \draw[->] (3.75,-0.25) -- (3.75,-1.25);
      \node[anchor=west] at (3.75,-0.75) {$F_{u,v}$};
    \end{tikzpicture}
  \end{center}
  \caption{\label{fig:bijectivity} The construction employed by the proof of
  Theorem~\ref{thm:bijection}. The patterned background denotes the blocked
  cells. For each value of $u$ and $v$ the function $F_{u,v}$ is a bijection
  from $\F_2^k$ to $\F_2^k$ where $k$ is the length of $b$.}
\end{figure*}
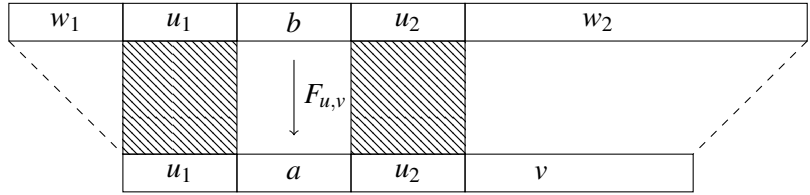

Recall that the \emph{reverse} of a vector $x = (x_0, \ldots, x_{n-1})$ is the
vector $x^R = (x_{n-1}, \ldots, x_0)$ with all components of $x$ appearing in
reverse order. Also, the \emph{complement} of $x$ is the vector
$x^C = (1 \oplus x_0, \ldots, 1 \oplus x_{n-1})$ where all components of $x$
appear negated. Given a local rule $f: \F_2^{\ell + r + 1} \to \F_2$ it is
possible to define its reverse $f^R : \F_2^{\ell + r + 1} \to \F_2$ as
$f^R(x) = f(x^R)$ and its complement $f^C: \F_2^{\ell + r + 1} \to \F_2$ as
$f^C(x) = 1 \oplus f(x)$ for all $x \in \F_2^{\ell + r + 1}$. The definition of
reverse and complement can also be extended to a CA $F: \F_2^n \to \F_2^m$ in
the following way:
\begin{align*}
  & F^R(x)_i = (F(x^R)^R)_i = f(x_{i+r}, \ldots, x_{i - \ell}) & & \forall 0 \le i < m\\
  & F^C(x)_i = 1 \oplus F(x)_i = 1 \oplus f(x_{i - \ell}, \ldots, x_{i + r}) && \forall 0 \le i < m \;.
\end{align*}

We can now show that, for a given $(t, n)$-AI CA it is possible to obtain other
(not necessarily distinct) $(t, n)$-AI by taking either its reverse or its
complement.

\begin{proposition}
  \label{lm:inv-refl}
  Let $F: \F_2^n \to F_2^m$ be a $(t, n)$-AI CA for some $n, m , t \in \N$ with
  $n = m + r + \ell$ and $r = \ell$. Then its reverse $F^R$ is also a
  $(t, n)$-AI CA.
\end{proposition}
\begin{proof}
  Starting with the reverse CA, by definition $F^R(x)$ is $F(x^R)^R$. Hence,
  given a set of indices $I$ with $|I| \le t$, the reflection of the $|I|$-ACA
  $\tilde{F}_I^R$ is:
  \begin{equation}
    \label{eq:refl-aca}
    \tilde{F}_I^R(x)_i = (\tilde{F}_J(x^R)^R)_i =
    \begin{cases}
      f(x_{i+r},\ldots,x_{i-\ell}) & \text{if $i \notin J$}\\
      x_i & \text{if $i \in J$}
    \end{cases}
  \end{equation}
  Where $J \subseteq \set{-\ell, \ldots, m + r -1}$ is defined as a ``reverse''
  of the set $I$ of indices, that is $J = \set{m +r -\ell -1 -i : i \in
    I}$. Notice that $J \subseteq [m]$ in all cases only if $\ell = r$. This
  means that for every set $I$ of indices for $F^R$, the corresponding set $J$
  of indices in $F$ is still a valid one (i.e., a subset of $[m]$). Notice that
  since $f$ generates a $(t, n)$-AI CA and $|J| = |I| \le t$, the resulting ACA
  is still $(t, n)$-AI.  
\end{proof}

Notice that, in general, if a $(t,n)$-AI CA has memory $\ell$ and anticipation
$r$ with $\ell \ne r$, its reverse might not be a $(t,n)$-AI CA. In fact, since
center permutivity of the local rule is not preserved, this negates a condition
for asynchrony-immunity that, by Theorem~\ref{thm:central-permutivity}, is
necessary for large enough values of $t$ and $n$.

\begin{proposition}
  \label{lm:inv-compl}
  Let $F: \F_2^n \to F_2^m$ be a $(t, n)$-AI CA for some $n, m , t \in \N$. Then
  its complement $F^C$ is also a $(t, n)$-AI CA.
\end{proposition}
\begin{proof}
  Let $y \in \F_2^m$ be a configuration, $I \subseteq [m]$ with $|I| \le t$, and let $(F^C_I)^{-1}(y)$ be the set of preimages of $y$ under the function $F^C_I$. By definition, for each $x \in \F_2^n$, $F^C(x) = 1 \oplus F(x)$. Hence, the set $(F^C_I)^{-1}(y)$ is $\set{x : 1 \oplus F_I(x) = y}$, which is $\set{x : F_I(x) = 1 \oplus y}$ which corresponds to $F^{-1}_I(1 \oplus y)$. Since $F$ is a $(t, n)$-AI CA, and all $y$ ranges across all elements of $\F_2^m$ (and thus $1 \oplus y$ does the same), $F^{-1}_I$ is balanced and $(F^C_I)^{-1}$ is also balanced. Since this holds for every set $I$ of cardinality at most $t$, it follows that $F^C$ is also a $(t, n)$-AI CA, as required. 
\end{proof}
Upper bounds on the size of the search space could be derived using techniques from~\cite{Cattaneo1997a} w.r.t. to the set of transformations $F^R, F^C, F^{RC}, Id$, where $Id$ is the identity transformation.

\section{Search of AI Rules up to $5$ Variables}
\label{sec:exp}
In order to search for asynchrony immune rules having additional cryptographic properties, by Theorem~\ref{thm:central-permutivity} and Propositions~\ref{lm:inv-refl} and~\ref{lm:inv-compl} we only need to explore center-permutive rules under the equivalence classes induced by reflection and complement.

In our experiments, we fixed the number of output bits in the CA to $m=8$. Since we are considering only center-permutive rules, we tested only the smallest value of $t$ satisfying the hypothesis of Theorem~\ref{thm:central-permutivity}. The reason why we limited our analysis to these particular values is twofold. First, checking for asynchrony immunity is a computationally cumbersome task, since it requires to determine the output distribution of the $t$-ACA for all possible choices of at most $t$ blocked cells. Second, the sizes of vectorial Boolean functions employed as nonlinear components in several real-world cryptographic primitives is limited. A concrete example is given by AES~\cite{aes}, which employs a S-box with $8$ output bits.

Table~\ref{tab:ca-param} shows all CA parameters considered in our experiments from $3$ to $5$ input variables of the local rules, while keeping the value of output bits fixed to $m=8$. Recall that, since we need to consider only center permutive local rules, we do not need to explore the entire $\mathcal{B}_{\ell+r+1}$ space, but only the subset $\mathcal{C}_{\ell+r+1}$ having cardinality $2^{2^{\ell+r}}$.

\begin{table}[t]
  \centering
  \scriptsize
  \begin{tabular}{cccccc}
    \hline\noalign{\smallskip}
    $n$ & $\ell$ & $r$ & $t$ & $|\mathcal{B}_{\ell+r+1}|$ & $|\mathcal{C}_{\ell+r+1}|$\\
    \noalign{\smallskip}\hline
    \noalign{\smallskip}
    $10$ & $1$ & $1$ & $2$ & $256$ & $16$  \\
    $11$ & $1$ & $2$ & $3$ & $65536$ & $256$ \\
    $12$ & $2$ & $2$ & $4$ & $\approx 4.3 \cdot 10^9$ & $65536$ \\
    \noalign{\smallskip}
    \hline
    \smallskip
  \end{tabular}
  \caption{CA parameters for $m=8$ output bits.}
  \label{tab:ca-param}
\end{table}

We started our investigation by performing an exhaustive search among all CA rules with $\ell=r=1$ (that is, rules of $3$ variables), which are also known in the CA literature as \emph{elementary rules}. Up to reflection and complement, and neglecting the identity rule that is trivially AI for every length $n$ and order $t$, out of the $2^{2^3}=256$ elementary rules we found that only rule 60 is $(2,10)$--asynchrony immune. However, rule $60$ is not interesting from the cryptographic standpoint, since it is linear (its ANF being $x_2 \oplus x_3$).

We thus extended the search by considering all local rules of $4$ and $5$ input variables, according to the values of $\ell$ and $r$ reported in Table~\ref{tab:ca-param}.

For the case of $4$ variables, the search returned a total of $18$ rules satisfying $(3,11)$--asynchrony immunity, among which several of them were nonlinear. Table~\ref{tab:nbr4} reports the Wolfram codes of the discovered rules, along with their nonlinearity values and algebraic normal form. It can be observed that $12$ rules out of $18$ are nonlinear, but none of them is a bent function (since the nonlinearity value in this case would be $6$).

\begin{table*}[t]
  \centering
  \scriptsize
  \begin{tabular}{ccl|ccl}
    \hline\noalign{\smallskip}
    Rule & $Nl(f)$ & \multicolumn{1}{c|}{$f(x_0,x_1,x_2,x_3)$} & Rule & $Nl(f)$ & \multicolumn{1}{c}{$f(x_0,x_1,x_2,x_3)$} \\
    \noalign{\smallskip}\hline
    \noalign{\smallskip}
    13107 & 0 & $1 \oplus x_1$ & 14028 & 2 & $x_1 \oplus x_0x_3 \oplus x_2x_3 \oplus x_0x_2x_3$ \\
    13116 & 4 & $x_1 \oplus x_2 \oplus x_3 \oplus x_2x_3$ & 14643 & 2 & $1 \oplus x_1 \oplus x_0x_3 \oplus x_0x_2x_3$ \\
    13155 & 2 & $1 \oplus x_1 \oplus x_2 \oplus x_0x_2 \oplus x_2x_3 \oplus x_0x_2x_3$ & 14796 & 2 & $x_1 \oplus x_3 \oplus x_0x_3 \oplus x_0x_2x_3$ \\
    13164 & 2 & $x_1 \oplus x_0x_2 \oplus x_3 \oplus x_0x_2x_3$ & 15411 & 4 & $1 \oplus x_1 \oplus x_3 \oplus x_2x_3$ \\
    13203 & 2 & $1 \oplus x_1 \oplus x_0x_2 \oplus x_0x_2x_3$ & 15420 & 0 & $x_1 \oplus x_2$ \\
    13212 & 2 & $x_1 \oplus x_2 \oplus x_0x_2 \oplus x_3 \oplus x_2x_3 \oplus x_0x_2x_3$ & 15555 & 0 & $1 \oplus x_1 \oplus x_2 \oplus x_3$ \\
    13251 & 4 & $1 \oplus x_1 \oplus x_2 \oplus x_2x_3$ & 15564 & 4 & $x_1 \oplus x_2x_3$ \\
    13260 & 0 & $x_1 \oplus x_3$ & 26214 & 0 & $x_0 \oplus x_1$ \\
    13875 & 2 & $1 \oplus x_1 \oplus x_3 \oplus x_0x_3 \oplus x_2x_3 \oplus x_0x_2x_3$ & 26265 & 0 & $1 \oplus x_0 \oplus x_1 \oplus x_3$ \\
    \noalign{\smallskip}
    \hline
    \smallskip
  \end{tabular}
  \caption{List of $(3,11)$--asynchrony immune CA rules of neighborhood size 4.}
  \label{tab:nbr4}
\end{table*}

For $5$ variables, Table~\ref{tab:nbr5} reports the list of $(4,12)$-AI CA. One can see that in this case most of the asynchrony immune functions are nonlinear, and moreover two of them achieve the maximum nonlinearity allowed by the quadratic bound, which in this case is $12$.
\begin{table*}[t]
  \centering
  \scriptsize
  \begin{tabular}{ccl|ccl}
    \hline\noalign{\smallskip}
    Rule & $Nl(f)$ & \multicolumn{1}{c|}{$f(x_1,x_2,x_3,x_4,x_5)$} & Rule & $Nl(f)$ & \multicolumn{1}{c}{$f(x_1,x_2,x_3,x_4,x_5)$} \\
    \noalign{\smallskip}\hline
    \noalign{\smallskip}
    252691440  & 4 & $ x_3 \oplus x_4 \oplus x_2x_4 \oplus x_5 \oplus x_4x_5 \oplus $ & 3031741620 & 8 & $x_2 \oplus x_1x_2 \oplus x_3$ \\
               &   & $x_2x_4x_5$                                                     &            &   &                                \\
    252702960  & 0 & $x_3 \oplus x_5$ & 3035673780 & 6 & $x_2 \oplus x_1x_2 \oplus x_3 \oplus x_2x_5 \oplus x_1x_2x_5 \oplus$ \\
               &   &                  &            &   & $x_2x_4x_5 \oplus x_1x_2x_4x_5$ \\
    253678110  & 10 & $x_1 \oplus x_2 \oplus x_1x_2 \oplus x_3 \oplus x_2x_4 \oplus $ & 3537031890 & 8 & $x_1 \oplus x_1x_2 \oplus x_3$ \\
               &    & $x_4x_5 \oplus x_1x_4x_5 \oplus x_1x_2x_4x_5$                   &            &   &                                \\
    255652080  & 4 & $x_3 \oplus x_2x_5 \oplus x_4x_5 \oplus x_2x_4x_5$ & 3537035730 & 8 & $x_1 \oplus x_1x_2 \oplus x_3 \oplus x_4 \oplus x_2x_4 \oplus$  \\
               &   &                                                   &            &   & $x_4x_5 \oplus x_2x_4x_5$                                       \\
    264499440  & 4 & $x_3 \oplus x_5 \oplus x_2x_5 \oplus x_2x_4x_5$& 3539005680 & 2 & $x_3 \oplus x_1x_4x_5 \oplus x_1x_2x_4x_5$ \\
    267390960  & 0 & $x_3 \oplus x_4$ & 4027576500 & 6 & $x_2 \oplus x_1x_2 \oplus x_3 \oplus x_2x_4 \oplus x_1x_2x_4 \oplus$  \\
               &   &                  &            &   & $x_5 \oplus x_2x_5 \oplus x_1x_2x_5 \oplus x_4x_5 \oplus$ \\
               &   &                  &            &   & $x_2x_4x_5 \oplus x_1x_2x_4x_5$ \\
    267448560  & 8 & $x_3 \oplus x_4x_5$ & 4030525680 & 4 & $x_3 \oplus x_2x_5 \oplus x_2x_4x_5$ \\
    505290270  & 8 & $x_1 \oplus x_2 \oplus x_1x_2 \oplus x_3$ & 4031508720 & 6 & $x_3 \oplus x_5 \oplus x_2x_5 \oplus x_1x_2x_5 \oplus x_4x_5 \oplus$ \\
               &   &                                          &            &   & $x_2x_4x_5 \oplus x_1x_2x_4x_5$                                      \\
    505336350  & 8 & $x_1 \oplus x_2 \oplus x_1x_2 \oplus x_3 \oplus x_2x_4 \oplus $ & 4038390000 & 2 & $x_3 \oplus x_2x_5 \oplus x_1x_2x_5 \oplus x_2x_4x_5 \oplus$ \\
               &   & $x_2x_4x_5$                                                     &            &   & $x_1x_2x_4x_5$                                              \\
    509222490  & 4 & $x_1 \oplus x_3 \oplus x_2x_4 \oplus x_1x_2x_4$ & 4039373040 & 4   & $x_3 \oplus x_5 \oplus x_2x_5 \oplus x_4x_5 \oplus x_2x_4x_5$ \\
    517136850  & 12 & $x_1 \oplus x_1x_2 \oplus x_3 \oplus x_4 \oplus x_2x_4 \oplus$ & 4040348370 & 6 & $x_1 \oplus x_1x_2 \oplus x_3 \oplus x_1x_4x_5 \oplus x_1x_2x_4x_5$ \\
               &    & $x_4x_5$                                                      &            &   &                                                                    \\
    756994590  & 12 & $x_1 \oplus x_2 \oplus x_1x_2 \oplus x_3 \oplus x_2x_4 \oplus$ & 4042268400 & 6 & $x_3 \oplus x_1x_4 \oplus x_2x_4 \oplus x_1x_2x_4 \oplus$ \\
               &    & $x_4x_5$                                                      &            &   & $x_1x_4x_5 \oplus x_2x_4x_5 \oplus x_1x_2x_4x_5$                                            \\
    2018211960 & 8 & $x_1x_2 \oplus x_3 \oplus x_5 \oplus x_2x_5 \oplus$ & 4042276080 & 4 & $x_3 \oplus x_2x_4 \oplus x_2x_4x_5$ \\
               &   & $x_4x_5 \oplus x_2x_4x_5$                           &            &   &                                     \\
    2018212080 & 10 & $x_3 \oplus x_1x_2x_4 \oplus x_5 \oplus x_2x_5 \oplus$ & 4042310640 & 4 & $x_3 \oplus x_4 \oplus x_2x_4 \oplus x_4x_5 \oplus x_2x_4x_5$ \\
               &    & $x_1x_2x_5 \oplus x_4x_5 \oplus x_2x_4x_5 \oplus$ &    &   &                                                                           \\
               &    & $x_1x_2x_4x_5$                                    &    &   &                                                                           \\
    2526451350 & 0 & $x_1 \oplus x_2 \oplus x_3$ & 4042318320 & 2 & $x_3 \oplus x_4 \oplus x_1x_4 \oplus x_2x_4 \oplus x_1x_2x_4 \oplus$ \\
               &   &                             &            &   & $x_4x_5 \oplus x_1x_4x_5 \oplus x_2x_4x_5 \oplus x_1x_2x_4x_5$       \\
    3023877300 & 6 & $x_2 \oplus x_1x_2 \oplus x_3 \oplus x_1x_2x_5 \oplus$ & 4042322160 & 0 & $x_3$ \\
               &   & $x_1x_2x_4x_5$                                         &            &   &       \\
    3027809460 & 8 & $x_2 \oplus x_1x_2 \oplus x_3 \oplus x_2x_5 \oplus$ \\
               &   & $x_2x_4x_5$                                         \\
    \noalign{\smallskip}
    \hline
    \smallskip
  \end{tabular}
  \caption{List of $(4,12)$--asynchrony immune CA rules of neighborhood size 5.}
  \label{tab:nbr5}
\end{table*}

\section{Open Problems}
\label{sec:open-problems}

There are many possible research directions for exploring asynchrony immune CA, mainly related to generalizations and relations with other models.

From the generalization point of view, we can relax the assumption that an attacker can control the updating of at most $t$ cells on $n$ cells CA. We can suppose that additional ``anti-tamper'' measures are present and, for example, that the attacker can only take control of non-consecutive cells. More in general, we can define $(\familyF,n)$-asynchrony immune CA where $\familyF \subseteq 2^{[m]}$ is a family of subsets of $\set{0,\ldots,m-1}$. The standard $(t,n)$-AI CA can be recovered by taking $\familyF$ as the set of all subsets of $[m]$ with cardinality at most $t$. It would be interesting to understand for what families of sets the theorems of this paper still hold. Also, what are some families that are ``plausible'' from a real-world point of view? This study will also require to explore the different methods that can be employed by an attacker to take control of some cells and what physical limits restrict the patterns of blocked cells that can be generated.

Another research direction is to find relations with already existing CA models that can be used to implement AI CA. Take, for example, the \emph{Multiple Updating Cycles CA} (MUCCA)~\cite{ManzoniEtAl_AFCA2016}, where each cell has a speed $1/k$ for a positive $k \in \N$ and a cell updates only if the current time step is a multiple of $k$. This means that, at different time steps, different cells might be active. If the current time step is not known or if it is under the attacker's control, then a CA that is $(t, n)$-AI can withstand any situation in which the number of ``slow'' cells (i.e., with speed less than $1$) is bounded by $t$. More generally, in what other models of ACA being asynchrony immune can protect from an attacker that controls some variables (like the time step in MUCCA)?

Subsequently, we have found that for size $n=11$ there are no $(11, 4)$-AI CA rules reaching maximum nonlinearity, that is, none of them is a bent function. Hence, an interesting question would be if there exists at least one bent AI CA rule of larger number of variables, and if it is possible to design an infinite family of bent AI CA.

Finally, from the cryptanalysis point of view, it would be interesting to analyze the resistance to clock-fault attacks of cryptographic primitives and ciphers based on cellular automata, such as the stream cipher CAR30~\cite{das13}, the $\chi$ S-box employed in the Keccak sponge construction~\cite{bertoni2011a}, or the CA-based S-boxes optimized through \emph{Genetic Programming} in~\cite{picek2017,mariot2019} and to verify if plugging in their design one of the AI CA rules found here decreases their possible vulnerability.

\bibliographystyle{spmpsci}
\bibliography{bibliography.bib}

\begin{thebibliography}{10}
\providecommand{\url}[1]{{#1}}
\providecommand{\urlprefix}{URL }
\expandafter\ifx\csname urlstyle\endcsname\relax
  \providecommand{\doi}[1]{DOI~\discretionary{}{}{}#1}\else
  \providecommand{\doi}{DOI~\discretionary{}{}{}\begingroup
  \urlstyle{rm}\Url}\fi

\bibitem{bertoni2011a}
Bertoni, G., Daemen, J., Peeters, M., Van~Assche, G., Van~Keer, R.: The {\sc
  keccak} reference (2008).
\newblock \urlprefix\url{http://keccak.team}

\bibitem{carlet2010a}
Carlet, C.: Boolean functions for cryptography and error correcting codes.
\newblock Boolean models and methods in mathematics, computer science, and
  engineering \textbf{2}, 257--397 (2010)

\bibitem{carlet2010b}
Carlet, C.: {Vectorial Boolean functions for cryptography}.
\newblock {Boolean models and methods in mathematics, computer science, and
  engineering} \textbf{134}, 398--469 (2010)

\bibitem{Cattaneo1997a}
Cattaneo, G., Formenti, E., Margara, L., Mauri, G.: Transformations of the
  one-dimensional cellular automata rule space.
\newblock Parallel Computing \textbf{23}(11), 1593 -- 1611 (1997)

\bibitem{chor1985a}
Chor, B., Goldreich, O., Hasted, J., Freidmann, J., Rudich, S., Smolensky, R.:
  The bit extraction problem or t-resilient functions.
\newblock In: Foundations of Computer Science, 26th Annual Symposium on, pp.
  396--407. IEEE (1985)

\bibitem{das13}
Das, S., Chowdhury, D.R.: {CAR30:} {A} new scalable stream cipher with rule 30.
\newblock {Cryptography and Communications} \textbf{5}(2), 137--162 (2013)

\bibitem{lili128}
Dawson, E., Clark, A., Golic, J., Millan, W., Penna, L., Simpson, L.: The
  lili-128 keystream generator.
\newblock In: Proceedings of first NESSIE Workshop (2000)

\bibitem{formenti2014a}
Formenti, E., Imai, K., Martin, B., Yun{\`e}s, J.B.: Advances on random
  sequence generation by uniform cellular automata.
\newblock In: Computing with New Resources, pp. 56--70. Springer (2014)

\bibitem{hoch04}
Hoch, J.J., Shamir, A.: Fault analysis of stream ciphers.
\newblock In: Cryptographic Hardware and Embedded Systems - {CHES} 2004: 6th
  International Workshop Cambridge, MA, USA, August 11-13, 2004. Proceedings,
  pp. 240--253 (2004)

\bibitem{leporati2014a}
Leporati, A., Mariot, L.: Cryptographic properties of bipermutive cellular
  automata rules.
\newblock Journal of Cellular Automata \textbf{9}, 437--475 (2014)

\bibitem{ManzoniEtAl_AFCA2016}
Manzoni, L., Porreca, A.E., Umeo, H.: {The Firing Squad Synchronization Problem
  on Higher-dimensional CA with Multiple Updating Cycles}.
\newblock In: 4th International Workshop on Applications and Fundamentals of
  Cellular Automata - AFCA 2016. Hiroshima, Japan (2016)

\bibitem{mariot2016}
Mariot, L.: Asynchrony immune cellular automata.
\newblock In: Cellular Automata - 12th International Conference on Cellular
  Automata for Research and Industry, {ACRI} 2016, Fez, Morocco, September 5-8,
  2016. Proceedings, pp. 176--181 (2016)

\bibitem{mariot2019}
Mariot, L., Picek, S., Leporati, A., Jakobovic, D.: Cellular automata based
  {S}-boxes.
\newblock Cryptography and Communications \textbf{11}(1), 41--62 (2019)

\bibitem{martin2006a}
Martin, B.: {A Walsh exploration of elementary CA rules}.
\newblock In: International Workshop on Cellular Automata, pp. 25--30.
  Hiroshima University (2006)

\bibitem{aes}
NIST/ITL/CSD: {{Advanced Encryption Standard (AES). FIPS PUB 197}} (2001).
\newblock \url{http://csrc.nist.gov/publications/fips/fips197/fips-197.pdf}

\bibitem{picek2017}
Picek, S., Mariot, L., Yang, B., Jakobovic, D., Mentens, N.: Design of
  {S}-boxes defined with cellular automata rules.
\newblock In: Proceedings of the Computing Frontiers Conference, CF'17, Siena,
  Italy, May 15-17, 2017, pp. 409--414 (2017)

\end{thebibliography}

\end{document}